\documentclass[3p,twocolumn]{elsarticle}

\usepackage{amsmath, amssymb, amsfonts, amsthm}
\usepackage{enumerate}
\usepackage{thmtools}
\usepackage{thm-restate}

\newtheorem{lemma}{Lemma}

\theoremstyle{definition}

\newcommand{\LZtrip}{LZ\textsubscript{3}}
\newcommand{\Znon}{\hat{z}}
\newcommand{\Zov}{z}

\journal{Information Processing Letters}

\begin{document}

\begin{frontmatter}
\title{Comparison of LZ77-type Parsings}
\author[DK]{Dmitry Kosolobov}\ead{dkosolobov@mail.ru}
\author[AS]{Arseny M. Shur}\ead{arseny.shur@urfu.ru}
\address[DK]{University of Helsinki, Helsinki, Finland}
\address[AS]{Ural Federal University, Ekaterinburg, Russia}

\begin{abstract}
We investigate the relations between different variants of the LZ77 parsing existing in the literature. All of them are defined as greedily constructed parsings encoding each phrase by reference to a string occurring earlier in the input. They differ by the phrase encodings: encoded by pairs (length + position of an earlier occurrence) or by triples (length + position of an earlier occurrence + the letter following the earlier occurring part); and they differ by allowing or not allowing overlaps between the phrase and its earlier occurrence. For a given string of length $n$ over an alphabet of size $\sigma$, denote the numbers of phrases in the parsings allowing (resp., not allowing) overlaps by $\Zov$ (resp., $\Znon$) for ``pairs'', and by $\Zov_3$ (resp., $\Znon_3$) for ``triples''. We prove the following bounds and provide series of examples showing that these bounds are tight:\\
\indent $\bullet$ $\Zov \le \Znon \le \Zov \cdot O(\log\frac{n}{\Zov\log_\sigma\Zov})$ and $\Zov_3 \le \Znon_3 \le \Zov_3 \cdot O(\log\frac{n}{\Zov_3\log_\sigma\Zov_3})$;\\
\indent $\bullet$ $\frac{1}2\Znon < \Znon_3 \le \Znon$ and $\frac{1}2\Zov < \Zov_3 \le \Zov$.
\end{abstract}
\begin{keyword}
LZ77 \sep lossless data compression \sep greedy parsing \sep non-overlapping phrases
\end{keyword}

\end{frontmatter}

\section{Introduction}

The Lempel--Ziv parsing~\cite{LZ77} (LZ77 for short) is one of the central techniques in the data compression and it plays an important role in stringology and algorithms in general. The literature on LZ77 is full of different variations of the parsing originally described by Lempel and Ziv~\cite{LZ77} (curiously, the most popular modern LZ77 modifications differ from the original one\footnote{The original parsing is the \LZtrip~parsing defined below.}). Some of these LZ77-based parsings lie at the heart of common compressors such as {\tt gzip}, {\tt 7-zip}, {\tt pkzip}, {\tt rar}, etc.~and some serve as a basis for compressed indexes on highly repetitive data (e.g., see \cite{GGKNP2,KreftNavarroTCS,nav2004}).

Most LZ77 variations have a noticeable optimality property: they have the least number of phrases among all reference-based parsings with the same fixed-length coding scheme for phrases (for details, see~\cite{Rytter03,LZ76} or Lemma~\ref{LZoptimal} below). The analysis in~\cite{CharikarEtAl} shows that many other popular reference-based methods (including LZ78~\cite{LZ78}) are significantly worse than LZ77 in the worst case. Probably, because of these ``near-optimal'' properties of LZ77, many authors often implicitly consider different LZ77 variations as somehow equivalent in terms of the number of produced phrases. Despite the fact that numerous works have been published in the last 40 years on this topic (e.g., see~\cite{Salomon} and references therein), to our knowledge, until very recently (see~\cite{GagiePrezzaNavarro,KempaPrezza}), there were no theoretical comparative studies of this side of LZ77 modifications.
We partially close this gap establishing tight bounds on the ratios between the numbers of phrases in several popular LZ77 variations. Note that the comparison of the parsings in terms of the bit size of their \emph{variable-length} encodings is a different and, as it seems, more challenging problem (see~\cite{FerraginaNittoVenturini,KosolobovLZ77enc}).

We investigate the relations between the most popular variants of the LZ77 parsing that one might find in the existing literature on the subject. All of them are defined as greedily constructed parsings that encode each phrase by reference to a string occurring earlier in the input, but they differ by the format of the phrase encodings and by the constraints imposed on earlier phrase occurrences. We primarily investigate four LZ77 variants that, at a generic step of the left-to-right greedy construction, define the phrase $f$ starting at the current position $i$ as follows:
\begin{enumerate}
\item $f$ is the longest string that starts at position $i$ and occurs at position $j < i$ (or $f$ is a letter if such string is empty);
\item as in 1, but $j \le i - |f|$;
\item $f$ is the shortest string that starts at position $i$ and does not have occurrences at positions $j < i$ (but $f$~can occur earlier if it is the last phrase in the parsing);
\item as in 3, but $j \le i - |f| + 1$.
\end{enumerate}

We call these parsings, respectively, \emph{LZ parsing}, \emph{non-overlapping LZ (novLZ) parsing}, \emph{\LZtrip~parsing}, and \emph{non-overlapping \LZtrip~(nov\LZtrip) parsing} (formal definitions are given below). For a given string of length $n$ over an alphabet of size $\sigma$, denote the numbers of phrases in thus defined parsings by, respectively, $\Zov, \Znon, \Zov_3, \Znon_3$. The non-one-letter phrases of LZ and novLZ parsings can be encoded by pairs of integers: the length of $f$ plus the offset ($i - j$) to an earlier occurrence of $f$. The phrases of \LZtrip~and nov\LZtrip~parsings can be encoded by triples (hence the subscript ``$3$''): the length of $f$ plus the offset ($i - j$) to an earlier occurrence of $f[1..|f|{-}1]$ plus the letter $f[|f|]$. We prove that the numbers of phrases in the considered LZ77 parsings are related as follows.\footnote{Throughout the paper, all logarithms have base $2$ if it is not explicitly stated otherwise.}

\begin{restatable}{thm}{thmmain}
For any given string of length $n$ over an alphabet of size $\sigma$, one has $\Zov \le \Znon \le \Zov \cdot O(\log\frac{n}{\Zov\log_\sigma\Zov})$ and $\Zov_3 \le \Znon_3 \le \Zov_3 \cdot O(\log\frac{n}{\Zov_3\log_\sigma\Zov_3})$.\label{MainTheorem}
\end{restatable}

The simpler bound $\Zov \le \Znon \le \Zov \cdot O(\log\frac{n}{\Zov})$ is easily implied by known results (e.g., by~\cite[Lem.~8]{Gawrychowski}) but our upper bound is better; in fact, it is tight, as the following theorem shows.

\begin{restatable}{thm}{thmmainexample}
For any integers $n > 1$, $\sigma \in [2..n]$, $\Zov \in [\sigma .. \frac{n}{\log_\sigma n}]$, there is a string of length $n$ over an alphabet of size $\sigma$ such that the sizes of its LZ and novLZ parsings are, respectively, $\Theta(\Zov)$ and $\Omega(\Zov\log\frac{n}{\Zov\log_\sigma\Zov})$. The same result holds for the \LZtrip/nov\LZtrip~parsings.
\label{MainExampleTheorem}
\end{restatable}

Note that while the necessity of the condition $\Zov \ge \sigma$ in this theorem is obvious, the condition $\Zov \le \frac{n}{\log_\sigma n}$ is justified by the well-known fact that the size of the LZ/\LZtrip\ parsing of any string of length $n$ over an alphabet of size $\sigma$ is at most $O(\frac{n}{\log_\sigma n})$ (see~\cite[Th. 2]{LZ76}).

Theorems~\ref{MainTheorem} and~\ref{MainExampleTheorem} are the main results of this paper. To complete the picture, we also investigate the relations between the numbers $\Zov$, $\Zov_3$ and, respectively, $\Znon$, $\Znon_3$, proving simple bounds and their tightness in the following theorem.

\begin{restatable}{thm}{thmpairsvstriples}
For any given string, one has $\frac{1}2\Zov < \Zov_3 \le \Zov$ and $\frac{1}2\Znon < \Znon_3 \le \Znon$. These bounds are tight since, for each $k\ge 1$ and each of the four restrictions $\Zov_3=\Zov=k$; $\Zov_3=k$ and $\Zov=2k-1$; $\Znon_3=\Znon=k$; $\Znon_3=k$ and $\Znon=2k-1$ there is a binary string satisfying this restriction.
\label{PairsVsTriplesTheorem}
\end{restatable}

It is known that a random string of length $n$ has $\Theta(n/\log_\sigma n)$ phrases in its Lempel--Ziv parsings (see~\cite[Th. 3]{LZ76}). A ``reasonably compressible'' string has, say, $\Omega(n/\log^{O(1)} n)$ phrases. For these strings our theorems imply that the sizes of all four considered LZ77 parsings are within $O(\log\log n)$ factor from each other; thus, we partially support the intuition that all these LZ77 variations are similar.

The paper is organized as follows. In Section~\ref{SectPrelim} we formalize the definitions of the LZ77 parsings under consideration and introduce some useful tools. In Section~\ref{SectOverlapVsNonoverlap}, the proofs of the main results (Theorems~\ref{MainTheorem} and~\ref{MainExampleTheorem}) are given. Theorem~\ref{PairsVsTriplesTheorem} is proved in Section~\ref{SectLZvsLZ3}. We conclude with some remarks and open problems in Section~\ref{SectConclusion}.

\section{Preliminaries}\label{SectPrelim}

A \emph{string $s$ of length $n$} over an alphabet $\Sigma$ is a map $\{1,2,\ldots,n\} \mapsto \Sigma$, where $n$ is referred to as the \emph{length of $s$}, denoted by $|s|$. We write $s[i]$ for the $i$th letter of $s$ and $s[i..j]$ for $s[i]s[i{+}1]\cdots s[j]$.
A string $u$ is a \emph{substring} of $s$ if $u = s[i..j]$ for some $i$ and $j$; the pair $(i,j)$ is not necessarily unique and we say that $i$ specifies an \emph{occurrence} of $u$ in $s$. A substring $s[1..j]$ (resp., $s[i..n]$) is a \emph{prefix} (resp. \emph{suffix}) of $s$. For any $i,j$, the set $\{k\in \mathbb{Z} \colon i \le k \le j\}$ (possibly empty) is denoted by $[i..j]$. A \emph{decomposition} of a string is its representation as the concatenation of nonempty substrings; writing a decomposition, we separate these substrings by dots. Two strings $u$ and $v$ are called \emph{conjugate} if $u = xy$ and $v = yx$ for some $x$ and $y$. An integer $p \in [1..|s|]$ is called a \emph{period} of $s$ if $s[i] = s[i{+}p]$ for any $i \in [1..|s|{-}p]$. The following lemma is obvious.

\begin{lemma}
Suppose that, in a string $s$, we have $w = s[i..j] = s[i'..j']$ and $i < i' \le j$; then $i' - i$ is a period of $w$.\label{OverlapPeriod}
\end{lemma}

For a given string $s$, the \emph{LZ} (resp., \emph{novLZ}) \emph{parsing} of $s$ is the decomposition $s = f_1f_2\cdots f_r$ built from left to right by the following greedy procedure: if a prefix $s[1..i{-}1] = f_1f_2\cdots f_{p-1}$ is already processed, then the string $f_p$ (which is called a \emph{phrase}) is either the letter $s[i]$ that does not occur in $s[1..i{-}1]$ or is the longest string that starts at position $i$ and has an occurrence at position $j < i$ (resp., $j \le i - |f_p|$). The \emph{\LZtrip} (resp., \emph{nov\LZtrip}) \emph{parsing} is constructed by an analogous greedy procedure but the phrase $f_p$ is chosen as the longest string occurring at position $i$ such that the string $f_p[1..|f_p|{-}1]$ has an occurrence at position $j < i$ (resp., $j \le i - |f_p| + 1$).

Consider $s = abababc$. The LZ, novLZ, \LZtrip, and nov\LZtrip~parsings of $s$ are, respectively, $a.b.abab.c$, $a.b.ab.ab.c$, $a.b.ababc$, and $a.b.aba.bc$.

Let $s = t_1t_2\cdots t_r$ be a decomposition of $s$ into non-empty strings $t_1, \ldots, t_r$. We say that $t_1t_2\cdots t_r$ is an \emph{LZ-type} (resp., \emph{novLZ-type}) \emph{parsing} if for each $i \in [1..r]$, the string $t_i$ either is a letter or has an occurrence in the string $s[1..|t_1\cdots t_i|{-}1]$ (resp., in $t_1t_2\cdots t_{i-1}$). Analogously, we say that $t_1t_2\cdots t_r$ is an \emph{\LZtrip-type} (resp., \emph{nov\LZtrip-type}) \emph{parsing} if for each $i \in [1..r]$, the string $t_i[1..|t_i|{-}1]$ has an occurrence in the string $s[1..|t_1\cdots t_i|{-}2]$ (resp., in $t_1t_2\cdots t_{i-1}$).

The number of phrases in a parsing is called the \emph{size} of the parsing. We write $\Zov$ (resp., $\Znon$, $\Zov_3$, $\Znon_3$) to denote the size of the LZ (resp., novLZ, \LZtrip, nov\LZtrip) parsing of a given string.

Our main tool in the subsequent analysis is the following well-known optimality lemma (see \cite[Th. 1]{LZ76}). We omit the proof as it is straightforward.
\begin{lemma}
For any given string, the size of its LZ (resp., novLZ, \LZtrip, nov\LZtrip) parsing is less than or equal to the size of any LZ-type (resp., novLZ-type, \LZtrip-type, nov\LZtrip-type) parsing.\label{LZoptimal}
\end{lemma}

\section{Relations Between Overlapping and Non-overlapping Parsings}\label{SectOverlapVsNonoverlap}

For the proof of Theorem~\ref{MainTheorem}, we need the following technical lemma.

\begin{lemma}
Suppose that $t_1, \ldots, t_r$ is a sequence of positive numbers such that $t_1 + t_2 + \cdots + t_r \le n$ for some $n > 0$; then, for any given $k > 0$, we have $\sum_{i=1}^r \log\frac{t_i}{k} \le r\log\frac{n}{rk}$.\label{TechLemma}
\end{lemma}
\begin{proof}
Denote $\alpha_i = \frac{t_i}k$. Note that $\alpha_1 + \cdots + \alpha_r \le \frac{n}{k}$. A well-known corollary of the concavity of the function $\log$ is that the sum $\sum_{i=1}^r \log\alpha_i$ is maximized whenever all $\alpha_i$ are equal and maximal, i.e., $\alpha_i = \frac{n}{rk}$ for all $i \in [1..r]$. Hence, the result follows.
\end{proof}

\thmmain*

\begin{proof}
Let us consider the case of $\Zov$ and $\Znon$; the proof for $\Zov_3$ and $\Znon_3$ can be reconstructed by analogy.

Since the novLZ parsing of $s$ is an LZ-type parsing, $\Zov \le \Znon$ by Lemma~\ref{LZoptimal}. Hence, it suffices to prove that $\Znon \le \Zov\cdot O(\log\frac{n}{\Zov\log_\sigma\Zov})$. The idea of the proof is to use the LZ parsing $f_1f_2\cdots f_{\Zov}$ of $s$ to construct a novLZ-type parsing of size $\Zov\cdot O(\log\frac{n}{\Zov\log_\sigma\Zov})$; then, the required bound follows from Lemma~\ref{LZoptimal}.

We construct a new parsing for $s$ substituting each phrase $f_i$ with a set of new phrases. If a phrase $f_i$ has an occurrence in the string $f_1\cdots f_{i-1}$, then we do not alter $f_i$ and include it in the new parsing. Consider a phrase $f_i$ such that the leftmost occurrence of $f_i$ in the string $f_1\cdots f_i$ occurs at position $j$ such that $|f_1\cdots f_{i-1}| - |f_i| + 1 < j$ (i.e., this occurrence of $f_i$ overlaps with $f_i$). Let us choose an arbitrary constant $\alpha \in (0,1)$. Denote $k = \alpha\log_{\sigma}\Zov$. We first discuss how to process the case $j \le |f_1\cdots f_{i-1}| - k$ (i.e., when the leftmost occurrence of $f_i$ is farther than $k$ letters from $f_i$).

By Lemma~\ref{OverlapPeriod}, $p = |f_1\cdots f_{i-1}| + 1 - j$ is a period of $f_i$ and $p \in [k..|f_i|]$. We decompose $f_i$ as follows: $f_i = t_1\cdots t_r$, where $|t_1| = 2^0p, |t_2| = 2^1p, \ldots, |t_{r-1}| = 2^{r{-}2}p$, and $t_r$ is a non-empty suffix of $f_i$ of length ${\le}2^{r-1}p$. Since $p$ is a period of $f_i$ and the substring of length $p$ preceding the phrase $f_i$ is equal to $f_i[1..p]$, any string $t_h$ from the decomposition occurs at $2^{h-1}p$ positions to the left and, since $|t_h| = 2^{h-1}p$, this occurrence does not overlap $t_h$. Therefore, we can include the strings $t_1, \ldots, t_r$ from the decomposition $f_i = t_1\cdots t_r$ as phrases in the novLZ-type parsing under construction. It is easy to see that $r = O(\log\frac{|f_i|}{p})$. Since $p \ge k$, we obtain $r = O(\log\frac{|f_i|}{k})$. Hence, it follows from Lemma~\ref{TechLemma} that the number of new phrases introduced by all such decompositions is upper bounded by $O(\sum_{i=1}^{\Zov} \log\frac{|f_i|}{k}) \le \Zov\cdot O(\log\frac{n}{\Zov k}) \le \Zov\cdot O(\log\frac{n}{\Zov\log_\sigma\Zov})$, exactly as required.

Now we process each phrase $f_i$ whose leftmost occurrence is at position $j > |f_1\cdots f_{i-1}| - k$ and overlaps $f_i$. Again, $p = |f_1\cdots f_{i-1}| + 1 - j$ is a period of $f_i$. Denote $c = \lfloor k / p\rfloor$. Note that $\frac{k}2 \le cp \le k$. Suppose that $cp < |f_i|$. We decompose $f_i$ as $f_i = t_0t_1\cdots t_r$, where $|t_0| = cp, |t_1| = 2^0cp, |t_2| = 2^1cp, \ldots, |t_{r-1}| = 2^{r-2}cp$, and $t_r$ is a non-empty suffix of $f_i$ of length ${\le}2^{r-1}cp$. As in the above analysis, it is easy to show that, for each $h \in [1..r]$, the substring $t_h$ from the decomposition has a non-overlapping left occurrence and, therefore, we can include the strings $t_1, \ldots, t_r$ as phrases in the novLZ-type parsing under construction. If $t_0$ also has a non-overlapping left occurrence, we include $t_0$ in the parsing; otherwise, we further decompose $t_0$ into one letter phrases. Since $cp > \frac{k}{2}$, it follows from the same arguments as in the case $j \le |f_1\cdots f_{i-1}| - k$ that the substrings $t_1, \ldots, t_r$ from all such decompositions add at most $\Zov\cdot O(\log\frac{n}{\Zov\log_\sigma\Zov})$ phrases. Let us show that the substrings $t_0$ decomposed into letters add $o(\Zov)$ phrases (the substrings $t_0$ that have non-overlapping left occurrences, obviously, add at most $\Zov$ phrases).

The crucial observation is that the length of each substring $t_0$ is at most $k$ and there are only at most $k\sigma^k$ distinct strings of length at most $k$ in $s$. Therefore, at most $k\sigma^k$ substrings $t_0$ will be decomposed into letters and, thus, they in total add at most $k^2\sigma^k = \Zov^\alpha \alpha^2\log^2_\sigma\Zov$ one letter phrases, which is $o(\Zov)$, i.e., negligible compared to~$\Zov$.

In the remaining case $cp \ge |f_i|$, we simply decompose $f_i$ into $|f_i|$ one letter phrases. Due to the greedy nature of the LZ parsing, all  strings $f_j f_{j+1}[1]$ (a phrase plus the following letter), for $j \in [1..\Zov{-}1]$, are distinct. Hence, using a counting argument analogous to the above one, it can be shown that there are at most $k\sigma^k$ such $f_i$ with $cp \ge |f_i|$ and their decompositions add at most $o(\Zov)$ phrases.
\end{proof}

The lower bound $\Zov$ for $\Znon$ (resp., $\Zov_3$ for $\Znon_3$) is obviously tight since the overlapping and non-overlapping parsings coincide for any string having no overlaps, and such \emph{overlap-free} strings of any length exist for any non-unary alphabet (see~\cite{Thue12}). The following recursively defined family of strings gives another possible construction with $\Zov = \Znon$ and $\Zov_3 = \Znon_3$: $s_1 = a_1$ and $s_i = s_{i-1}s_{i-1}a_i$, for $i > 1$, where $a_i$ are distinct letters; each string $s_i$ has length $2^i - 1$, and its LZ and novLZ (resp., \LZtrip{} and nov\LZtrip) parsings coincide and have size $2i - 1 = 2\log(|s_i| + 1) - 1$ (resp., $i = \log(|s_i| + 1)$). Further, for $a_1a_2\cdots a_n$, we also obviously have $\Zov = \Znon$ and $\Zov_3 = \Znon_3$. Combining these two constructions, one can easily describe, for arbitrary given integers $n > 0$ and $k \in [2\log(n + 1) - 1 .. n]$, a string of length $n$ with $k = \Zov = \Znon$ (resp., $k = \Zov_3 = \Znon_3$).

Theorem~\ref{MainExampleTheorem} proves the tightness of the upper bound $\Zov\cdot O(\frac{n}{\Zov\log_\sigma\Zov})$ for $\Znon$ (and of the respective upper bound $\Zov_3\cdot O(\frac{n}{\Zov_3\log_\sigma\Zov_3})$ for $\Znon_3$).

\thmmainexample*

\begin{proof}
We describe such string only for LZ/novLZ; however, our construction can be used for \LZtrip/nov\LZtrip\ as well and the analysis is analogous, so we omit the details.

The example for an unlimited alphabet is easy (for simplicity, we assume here that $n$ is a multiple of $\sigma$): the string $a_1^{n/\sigma}a_2^{n/\sigma}\cdots a_\sigma^{n/\sigma}$, where $a_1,\ldots,a_\sigma$ are distinct letters, satisfies $\Zov=2\sigma$ and $\Znon=\Omega(\sigma\log\frac{n}{\sigma})=\Omega(\Zov\log\frac{n}{\Zov\log_\sigma\Zov})$.

We generalize this simple example for alphabets of restricted size $\sigma$ replacing each letter $a_i$ with a string of length $\Theta(\log_\sigma z)$. Let us describe the strings that serve as replacements. Denote $d = \lceil\log_\sigma\Zov\rceil$. In~\cite{Cohn} it was shown that all $\sigma^d$ possible strings of length $d$ over an alphabet of size $\sigma$ can be arranged in a sequence $v_1, v_2, \ldots, v_{\sigma^d}$ (called a \emph{\mbox{$\sigma$-ary} Gray code}~\cite{Cohn,Gray}) such that, for any $i \in [2..\sigma^d]$, the strings $v_{i-1}$ and $v_i$ differ in exactly one position. Moreover, we can choose such sequence so that $v_1 = b^d$, where $b$ is an arbitrarily chosen letter from the alphabet. The strings $u_i = ab^{d-1}v_i$, where $a$ is a letter that differs from $b$, serve as the replacements for $a_i$. The important property of $u_i$ is that no two distinct strings $u_i$ and $u_j$ are conjugates; this follows from the observation that conjugates must contain two occurrences of $ab^{d-1}$, while the only string $u_i$ with this property is $(ab^{d-1})^2$.

Suppose that $\Zov \le 8$. Since $\Omega(\Zov\log\frac{n}{\Zov\log_\sigma\Zov}) = \Omega(\log n)$ in this case (note that $\sigma \le \Zov \le 8$), the statement of the theorem can be easily proved using the example string $a^n$.  Now suppose that $\Zov > 8$. Denote $k =  \lfloor \Zov/8 \rfloor$.  Observe that $k \ge 1$ and $\sigma^d = \sigma^{\lceil\log_\sigma z\rceil} \ge z > k$. Our example is the following string:
$$
s = u_1^{\lfloor\frac{n}{2kd}\rfloor} u_2^{\lfloor\frac{n}{2kd}\rfloor} \cdots u_k^{\lfloor\frac{n}{2kd}\rfloor},
$$
which consists of $k$ ``blocks'' $u_i^{\lfloor\frac{n}{2kd}\rfloor}$. Since $|u_i| = 2d$, the length of $s$ is $\lfloor\frac{n}{2kd}\rfloor 2kd \le n$. We append enough letters $a$ to the end of $s$ to make the length equal to $n$; such modification does not affect the proof that follows, so, without loss of generality, we assume that $|s| = n$.

Since $\Zov \le \frac{n}{\log_\sigma n}$ and $k \le \Zov / 8$, we have $kd \le (\Zov / 8) \lceil\log_\sigma\Zov\rceil \le \frac{n}{8\log_\sigma n} \lceil\log_\sigma n\rceil \le n / 4$. Therefore, $\lfloor\frac{n}{2kd}\rfloor \ge 2$, i.e., each block $u_i^{\lfloor\frac{n}{2kd}\rfloor}$ consists of at least two copies of $u_i$.

The string $s$ has an LZ-type parsing with at most $4$ phrases per block: $a.b.b^{2d-2}.u_1^{\lfloor\frac{n}{2kd}\rfloor-1}.$ $.u'_2.c_2.u''_2.u_2^{\lfloor\frac{n}{2kd}\rfloor-1}.\cdots.u'_k.c_k.u''_k.u_k^{\lfloor\frac{n}{2kd}\rfloor-1}$, where $u'_i$ (resp., $u''_i$) is the longest common prefix (resp., suffix) of $u_i$ and $u_{i-1}$, and $c_i$ is a letter. This shows that the size of the LZ parsing of $s$ is at most $4k \le \frac{\Zov}{2}$. On the other hand, let us demonstrate that the size of the non-overlapping LZ parsing of $s$ is at least $k\log\lfloor\frac{n}{4kd}\rfloor$.

Consider, for $i > 1$, the leftmost occurrence of $u_i$ in $s$. It is inside $u_j^2$ or $u_{j-1}u_j$ for some $j\le i$. In the first case, the occurrence is a conjugate of $u_j$, implying $j=i$. In the second case, it is a conjugate of either $u_{j-1}$ or $u_j$ (since $u_{j-1}$ and $u_j$ differ in exactly one position and have the same length $2d$); so again $j=i$. This means that the first $u_i$ in the $i$th block cannot have non-overlapping left occurrences. Hence, this $u_i$ contains at least one border between phrases of the novLZ parsing. Moreover, such a phrase containing the suffix of this $u_i$ has length at most $4d$ since if it has length greater than $4d$, then the second string $u_i$ in the $i$th block has a copy at a distance of more than $4d$ symbols to the left and, thus, this copy is inside the first $i - 1$ blocks, which is impossible. Analogously, one can show that the next phrase has length at most $8d$, then $16d$, and so on until the phrase border inside or immediately before the first occurrence of $u_{i+1}$. Thus, we have proved that at least $\log\lfloor\frac{n}{4kd}\rfloor$ phrases are needed for each of the $k$ blocks, as required. Therefore, since $\frac{\Zov}{8} - 1 < k \le \frac{\Zov}{8}$, we obtain $\Znon = \Omega(\Zov\log\frac{n}{\Zov \log_\sigma\Zov})$, i.e., the upper bound for $\Znon$ is reached on the string $s$.
\end{proof}

\section{Relations Between Parsings with Pairs and Triples}\label{SectLZvsLZ3}

Now we prove Theorem~\ref{PairsVsTriplesTheorem}.

\thmpairsvstriples*

\begin{proof}
Let us consider $\Zov$ and $\Zov_3$; the analysis of $\Znon$ and $\Znon_3$ is the same.

Let $f_1f_2\cdots f_{\Zov_3}$ be the \LZtrip~parsing of a string $s$. It is immediate from the definitions that $f_1t_2t'_2\cdots t_{\Zov_3}t'_{\Zov_3}$, where $t_i = f_i[1..|f_i|{-}1]$ and $t'_i = f_i[|f_i|]$, is an LZ-type parsing of $s$ of size at most $2\Zov_3-1$ (we remove empty strings $t_i$ from the parsing). Hence $\Zov < 2\Zov_3$ by Lemma~\ref{LZoptimal}. Further, the LZ parsing of $s$ is an \LZtrip-type parsing of $s$ by definition. Therefore, again by Lemma~\ref{LZoptimal}, we obtain $\Zov_3 \le \Zov$.

Let us show the tightness of the bounds. The verification of examples presented below might be tedious but, nevertheless, is quite straightforward, so we do not discuss all details. However, to ensure that the constructions are correct, we wrote a computer program checking the examples for small parameters $k$.

Let $k\ge 2$ (the case $k=1$ is trivial). The restriction $\Zov=2k-1$, $\Zov_3=k$ is satisfied by the string
$$
aabaab^3aab^7\cdots aab^{2^{k-2}-1}aab^{2^{k-2}},
$$
whose  LZ and \LZtrip~parsings are, respectively,
$$
\begin{array}{l}
a.a.b.aab.b^2.aab^3.b^4.\cdots.aab^{2^{k-2}-1}.b,\\
a.ab.aabb.baab^4.b^3aab^8.\cdots .b^{2^{k-3}-1}aab^{2^{k-2}}.
\end{array}
$$
Next, the equalities $\Zov=\Zov_3=2k$ hold for the string
$$
abab^4abab^{10}\cdots abab^{3\cdot 2^{k-1}-2},
$$
having the following  LZ and \LZtrip~parsings:
$$
\begin{array}{l}
a.b.ab.b^3.abab^4.b^6.\cdots .abab^{3\cdot 2^{k-2}-2}.b^{3\cdot 2^{k-2}},\\
a.b.ab^2.b^2a.bab^5.b^5a.\cdots. bab^{3\cdot 2^{k-2}-1}.b^{3\cdot 2^{k-2}-1}.
\end{array}
$$
Note that if we delete the last $3\cdot 2^{k-2}$ $b$'s, both parsings of the resulting string will have size $2k-1$.  Therefore, the equality $\Zov=\Zov_3=k$ can be achieved for any $k$. Further, the string
$$
a^2ba^5b^3a^{11}b^7\cdots a^{3\cdot 2^{k-1}-1}b^{2^k-1}a^{3\cdot 2^k-1}b^{2^k}
$$
satisfies $\Znon=4k+3$, $\Znon_3=2k+2$ as its corresponding novLZ and nov\LZtrip~parsings look as follows:
$$
\begin{array}{l}
a.a.b.a^2\!.a^2\!.ab.b.ba^5\!.a^5\!.ab^3\!.b^3\!.ba^{11}\!.a^{11}\!.ab^7.\cdots.\\
\hfill ba^{3\cdot 2^{k-1}-1}\!.a^{3\cdot 2^{k-1}-1}\!.ab^{2^k-1}\!.b,\\
a.ab.a^3\!.a^2b^2\!.ba^6\!.a^5b^4\!.\cdots .b^{2^{k-1}-1}a^{3\cdot 2^{k-1}}\!\!.a^{3\cdot 2^{k-1}-1}b^{2^k}\!.
\end{array}
$$
If we delete the last phrase of the nov\LZtrip~parsing, the resulting string will satisfy $\Znon=4k+1$, $\Znon_3=2k+1$. Therefore, the condition $\Znon=2k-1$, $\Znon=k$ can be satisfied for any $k$. Finally, it is easy to verify that one has $\Znon_3 = \Znon =k$ for the string $(ab)^{2^{k-2}}$. The theorem is proved.
\end{proof}

\section{Concluding Remarks}\label{SectConclusion}

In the literature there is still a lack of information concerning the relations between different measures of compressibility for highly repetitive texts. In this paper we investigated the relations between the most popular versions of LZ77 but, besides LZ77, there are other popular measures. For instance, it is a major open problem to find tight relations between an LZ77 parsing of a given string and the number of runs in its Burrows--Wheeler transform (see~\cite{GagiePrezzaNavarro}). Further, it is known that the size of the smallest grammar of any given string of length $n$ is within $O(\log n)$ factor of the size of the LZ parsing and it is known that this bound is tight to within a factor $O(\log\log n)$ (see~\cite{BilleGagieGortzPrezza, CharikarEtAl, Jez2, Rytter03}); but it is still open whether this bound can be improved to $O(\frac{\log n}{\log\log n})$. The things are not always clear even in the realm of LZ77-alike parsings: for example, it is still not known whether, as it was conjectured in~\cite{KreftNavarroTCS}, the so-called LZ-End parsing contains at most $2\Znon$ phrases. Finally, note a rather unexpected connection between $\Znon$ and the number of distinct factors in the Lyndon decomposition of a string \cite{Karkkainen_etalSTACS2017}.

We refer the reader to~\cite{GagiePrezzaNavarro} and~\cite{KempaPrezza} and references therein for further discussion on different measures of compressibility and their relations; other compression schemes and results on their relations can also be found in~\cite{StorerSzymanski}.

\subparagraph{Acknowledgement.}
The authors would like to thank the anonymous referee for detailed comments that helped to improve the paper.


\bibliography{lz-cmp}
\bibliographystyle{elsarticle-num-names-sort}
\end{document}